\documentclass{llncs}

\usepackage{amsmath,amssymb}

\usepackage{float}
\usepackage[dvips]{epsfig}
\usepackage{epsf}
\usepackage{url}
\usepackage{graphicx}

\def\amod#1 #2{#1\ ({\rm mod}\ #2)}

\frontmatter

\title{On the Number of Unbordered Factors}

\author{Daniel Go\v{c} \and Hamoon Mousavi \and Jeffrey Shallit}

\institute{School of Computer Science \\
University of Waterloo \\
Waterloo, ON  N2L 3G1 \\
Canada \\
\email{\{dgoc,hamoon.mousavihaji,shallit\}@cs.uwaterloo.ca}
}

\begin{document}

\maketitle

\begin{abstract}
We illustrate a general technique for enumerating factors of $k$-automatic
sequences by proving a
conjecture on the number $f(n)$ of unbordered factors of the Thue-Morse
sequence.  We show that $f(n) \leq n$ for $n \geq 4$ and that
$f(n) = n$ infinitely often.
We also give examples of automatic sequences having exactly $2$ unbordered
factors of every length.
\end{abstract}

\section{Introduction}

In this paper, we are concerned with certain factors
of $k$-automatic sequences.  Roughly speaking, a sequence ${\bf x} =
a_0 a_1 a_2 \cdots$ over a finite alphabet
$\Delta$ is said to be $k$-automatic if there exists a
finite automaton that, on input $n$ expressed in base $k$, reaches
a state with output $a_n$.  Automatic sequences were popularized
by a celebrated paper of Cobham \cite{Cobham:1972} and have been
widely studied; see \cite{Allouche&Shallit:2003}.

More precisely, let $k$ be an integer $\geq 2$, and set
$\Sigma_k = \lbrace 0, 1, \ldots, k-1 \rbrace$.  Let
$M = (Q, \Sigma_k, \Delta, \delta, q_0, \tau)$ be a deterministic
finite automaton with output (DFAO) with transition function
$\delta:Q \times \Sigma_k \rightarrow Q$ and output function
$\tau:Q \rightarrow \Delta$.  Let $(n)_k$ denote the canonical base-$k$
representation of $n$, without leading zeros, and starting with the
most significant digit.  Then we say that $M$ generates the
sequence $(a_n)_{n \geq 0}$ if
$a_n = \tau(\delta(q_0, (n)_k))$ for all $n \geq 0$.

The prototypical example of a $k$-automatic sequence is the
Thue-Morse sequence ${\bf t} = t_0 t_1 t_2 \cdots = 01101001 \cdots$,
defined by the relations $t_0 = 0$ and
$t_{2n} = t_n$, $t_{2n+1} = 1-t_n$
for $n \geq 0$.    It is generated by the DFAO below in Figure~\ref{fig1}.

\begin{figure}[H]
\leavevmode
\def\epsfsize#1#2{1.0#1}
\centerline{\epsfbox{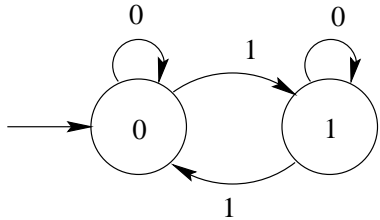}}
\protect\label{fig1}
\caption{A finite automaton generating the Thue-Morse sequence $\bf t$}
\end{figure}

A {\it factor} of the sequence $\bf x$ is a finite word of the
form $a_i \cdots a_j$.  A finite
word $w$ is said to be {\it bordered\/} if there
is some finite nonempty word $x \not= w$ that is both a prefix and
a suffix of $w$ \cite{Silberger:1971,Nielsen:1973,Ehrenfeucht&Silberger:1979,Silberger:1980}.  For example, the English word
{\tt ionization} is bordered, as it begins and ends with
{\tt ion}.  Otherwise $w$ is said to be {\it unbordered}.

Recently, there has been significant interest in the properties
of unbordered factors; see, for example,
\cite{Holub:2005,Harju&Nowotka:2007,Duval&Harju&Nowotka:2008,Holub&Nowotka:2010}.
In particular, Currie and Saari \cite{Currie&Saari:2009} studied 
the unbordered factors of the Thue-Morse word.

Currie and Saari \cite{Currie&Saari:2009} proved that if
$n \not\equiv \amod{1} {6}$, then
the Thue-Morse word has an unbordered factor of length $n$, but
left it open to decide for which lengths congruent to $1$ (mod $6$)
this property holds.  This was solved in
\cite{Goc&Henshall&Shallit:2012}, where the following characterization
is given:

\begin{theorem}
The Thue-Morse sequence $\bf t$
has an unbordered factor of length $n$ if and only
if $(n)_2 \not\in 1(01^*0)^* 10^*1$.  
\end{theorem}

A harder problem is to come up with an expression for the number
of unbordered factors of $\bf t$.    
In \cite{Charlier&Rampersad&Shallit:2011}, the second author and
co-authors made the following conjecture:

\begin{conjecture}
Let $f(n)$ denote the number of unbordered factors of length $n$
in $\bf t$, the Thue-Morse sequence.  Then 
$f$ is given by $f(0) = 1$, $f(1) = 2$, $f(2) = 2$,
and the system of recurrences
\begin{eqnarray}
f(4n+1) & = & f(2n+1) \nonumber \\
f(8n+2) & =  &f(2n+1)-8f(4n) + f(4n+3) + 4f(8n) \nonumber \\
f(8n+3) & =  &2f(2n) - f(2n+1) + 5f(4n) + f(4n+2) - 3f(8n) \nonumber \\
f(8n+4) & =  &-4f(4n) + 2f(4n+2) + 2f(8n) \nonumber \\
f(8n+6) & =  &2f(2n)-f(2n+1) + f(4n) + f(4n+2) + f(4n+3) -f(8n) \nonumber \\
f(16n) & =  &-2f(4n) + 3f(8n)  \label{idl} \\
f(16n+7) & = & -2f(2n) + f(2n+1) -5f(4n) + f(4n+2) +3f(8n) \nonumber \\
f(16n+8) & = & -8f(4n) + 4f(4n+2) + 4f(8n)  \nonumber \\
f(16n+15) & = & -8f(4n) + 2f(4n+3) + 4f(8n) + f(8n+7) . \nonumber
\end{eqnarray}
for $n \geq 0$.
\label{unbtm}
\end{conjecture}

This conjecture was obtained by computing a large number of values
of $f$ and then looking for possible
linear relations among subsequences
of the form $(f(2^i n + j))_{n \geq 0}$.

This system suffices to calculate $f$ efficiently,
in $O(\log n)$ arithmetic steps.

We now summarize the rest of the paper.  In Section~\ref{proof},
we prove Conjecture~\ref{unbtm}.  In Section~\ref{kreg}, we discuss
how to obtain relations like those above for a given $k$-regular
sequence.  In Section~\ref{growth} we discuss the growth rate of
$f$ in detail.  Finally, in Section~\ref{other}, we give examples
of other sequences with interesting numbers of unbordered factors.

\section{Proof of the conjecture}
\label{proof}

We now outline our computational proof of
Conjecture~\ref{unbtm}.

First, we
need a little notation.  We extend the notion of canonical base-$k$
representation of a single non-negative integer to tuples of such
integers.  For example, by $(m,n)_k$ we mean the unique word over
the alphabet $\Sigma_k \times \Sigma_k$ such that the projection
$\pi_1$ onto the first coordinate gives the base-$k$ representation
of $m$, and the projection $\pi_2$ onto the second co-ordinate gives
the base-$k$ representation of $n$, where the shorter representation is
padded with leading $0$'s, if necessary, so that the representations
have the same length.  For example,
$(43,17)_2 = [1,0] [0,1] [1,0] [0,1] [1,0] [1,1]$.

\begin{proof}

Step 1:  Using the ideas in \cite{Goc&Henshall&Shallit:2012}, we created
an automaton $A$ of 23 states that accepts the
language $L$ of all words $(n,i)_2$ such that there is
a ``novel'' unbordered factor of length $n$ in $\bf t$ beginning
at position $i$. Here ``novel'' means
that this factor does not previously appear in any position to the left.
Thus, the number of such words with first component equal to $(n)_2$
equals $f(n)$, the number of unbordered factors of $\bf t$ of length $n$.
This automaton is illustrated below in Figure~\ref{aut} (rotated to fit
the figure more clearly).

\bigskip

Step 2:  Using the ideas in \cite{Charlier&Rampersad&Shallit:2011},
we now know that $f$ is a $2$-regular sequence,
with a ``linear representation'' that can be deduced from the structure
of $A$.  This gives matrices $M_0, M_1$ of dimension $23$ and vectors
$v, w$ such that
$f(n) = v M_{a_1} \cdots M_{a_i} w$
where $a_1 \cdots a_i$ is the base-$2$
representation of $n$, written with the most significant digit first.
They are given below.

$$M_0 = {\scriptsize \left[ \begin{array}{ccccccccccccccccccccccc}
1&0&0&0&0&0&0&0&1&0&0&0&0&0&0&0&0&0&0&0&0&0&0 \\
0&0&1&0&0&0&0&0&0&1&0&0&0&0&0&0&0&0&0&0&0&0&0 \\
0&0&1&0&0&0&0&0&0&0&0&1&0&0&0&0&0&0&0&0&0&0&0 \\
0&0&0&0&0&1&0&0&0&0&0&1&0&0&0&0&0&0&0&0&0&0&0 \\
0&0&0&0&0&0&0&0&0&0&1&0&1&0&0&0&0&0&0&0&0&0&0 \\
0&0&0&1&0&0&0&0&0&0&0&1&0&0&0&0&0&0&0&0&0&0&0 \\
0&0&0&0&0&0&0&1&0&0&0&0&0&0&0&0&1&0&0&0&0&0&0 \\
0&0&0&0&0&0&0&0&0&0&0&0&0&0&0&1&0&0&1&0&0&0&0 \\
0&0&0&0&0&0&0&0&0&0&0&0&1&0&0&0&0&0&0&0&0&1&0 \\
0&0&0&0&0&1&0&0&0&0&0&1&0&0&0&0&0&0&0&0&0&0&0 \\
0&0&0&0&0&0&0&0&0&0&0&0&1&1&0&0&0&0&0&0&0&0&0 \\
0&0&1&0&0&0&0&0&0&0&0&1&0&0&0&0&0&0&0&0&0&0&0 \\
0&0&0&0&0&0&0&0&0&0&0&0&0&0&0&0&0&0&0&0&0&0&0 \\
0&0&0&0&0&0&0&0&0&0&1&0&1&0&0&0&0&0&0&0&0&0&0 \\
0&0&0&1&0&0&0&0&0&0&0&1&0&0&0&0&0&0&0&0&0&0&0 \\
0&0&0&0&0&0&0&0&0&0&0&0&0&0&0&0&0&0&0&0&0&0&0 \\
0&0&0&0&0&0&0&0&0&0&1&0&0&0&0&0&0&1&0&0&0&0&0 \\
0&0&0&0&0&0&0&0&0&0&0&0&0&0&0&1&0&0&1&0&0&0&0 \\
0&0&0&0&0&0&0&0&0&0&0&0&1&0&0&0&0&1&0&0&0&0&0 \\
0&0&0&0&0&0&0&0&0&0&0&0&1&0&0&0&0&0&0&0&1&0&0 \\
0&0&0&0&0&0&0&0&0&0&0&0&1&0&0&0&0&0&0&1&0&0&0 \\
0&0&0&0&0&0&0&0&0&0&0&0&0&0&0&0&0&0&0&0&0&0&0 \\
0&0&0&0&0&0&0&0&0&0&0&0&0&0&0&0&0&1&1&0&0&0&0 
\end{array}
\right]
}
$$

\begin{figure}[H]
\begin{center}
\includegraphics[width=8in,angle=90]{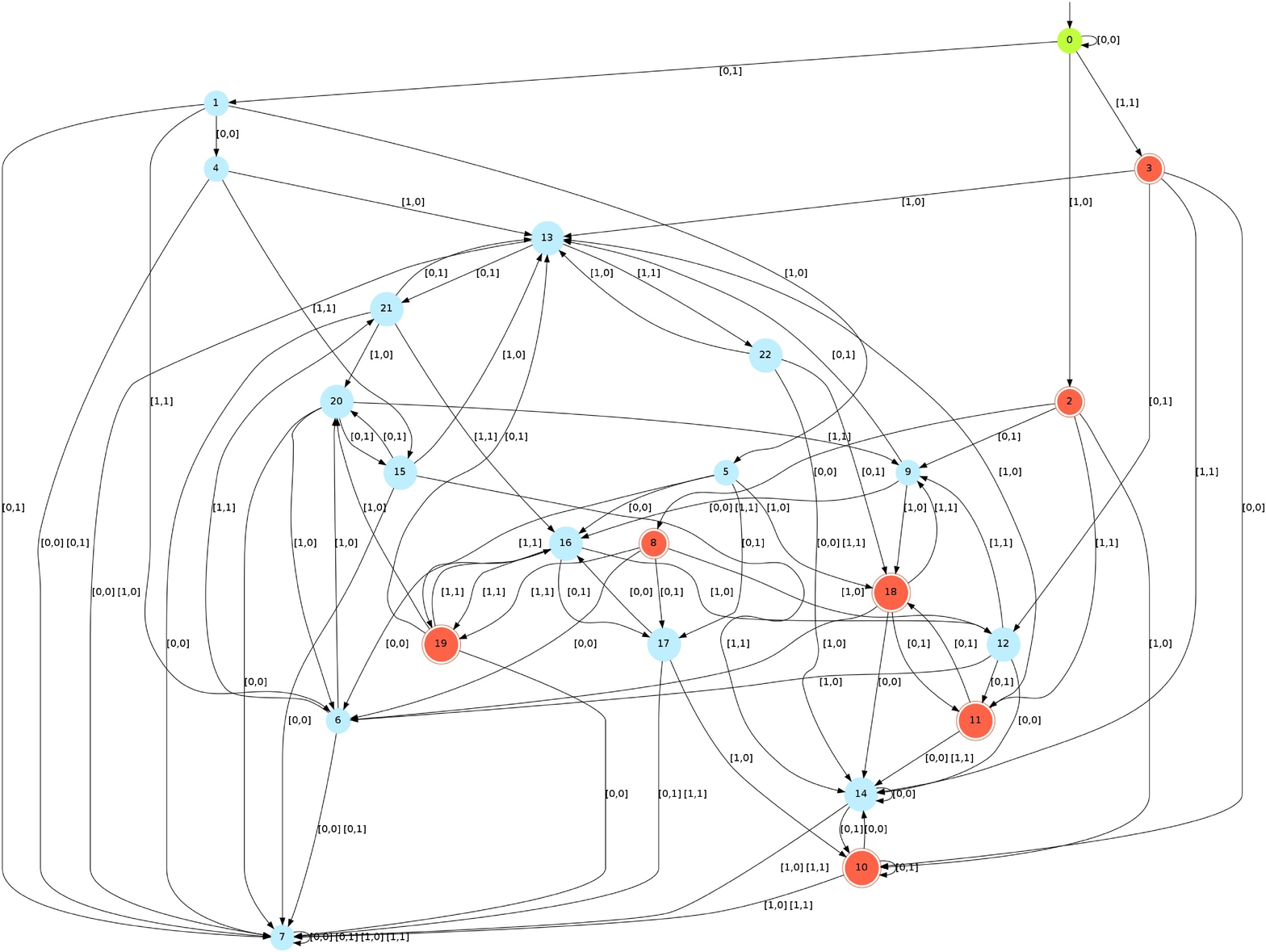}
\end{center}
\caption{Automaton accepting $(n,i)_2$ such that
there is a novel unbordered factor of length $n$ at position $i$ of
$\bf t$}
\label{aut}
\end{figure}

$$M_1 = {\scriptsize \left[ \begin{array}{ccccccccccccccccccccccc}
0&1&0&0&0&0&1&0&0&0&0&0&0&0&0&0&0&0&0&0&0&0&0 \\
0&0&0&0&0&0&0&0&0&0&1&1&0&0&0&0&0&0&0&0&0&0&0 \\
0&0&0&0&0&0&0&0&0&0&0&0&0&0&0&0&0&0&0&0&0&0&0 \\
0&0&0&0&0&0&0&0&0&0&0&0&0&0&0&1&1&0&0&0&0&0&0 \\
0&0&0&0&0&0&0&0&0&0&0&0&0&0&0&0&0&1&0&1&0&0&0 \\
0&0&0&0&0&0&0&0&0&0&1&1&0&0&0&0&0&0&0&0&0&0&0 \\
0&0&1&0&0&1&0&0&0&0&0&0&0&0&0&0&0&0&0&0&0&0&0 \\
0&0&0&0&1&0&0&0&0&1&0&0&0&0&0&0&0&0&0&0&0&0&0 \\
0&0&0&0&0&0&0&0&0&0&0&0&0&0&0&1&0&0&0&0&0&0&1 \\
0&0&0&0&0&0&0&0&0&0&0&0&0&0&0&1&1&0&0&0&0&0&0 \\
0&0&0&0&0&0&0&0&0&0&0&0&1&0&1&0&0&0&0&0&0&0&0 \\
0&0&0&0&0&0&0&0&0&0&0&0&0&0&0&0&0&0&0&0&0&0&0 \\
0&0&0&0&0&0&0&0&0&0&0&0&0&0&0&0&0&0&0&0&0&0&0 \\
0&0&0&0&0&0&0&0&0&0&0&0&0&0&0&0&0&1&0&1&0&0&0 \\
0&0&0&0&0&0&0&0&0&0&1&1&0&0&0&0&0&0&0&0&0&0&0 \\
0&0&0&0&0&0&0&0&0&0&0&0&0&1&0&0&0&0&0&1&0&0&0 \\
0&0&0&1&0&0&0&0&0&0&0&0&0&0&0&0&0&1&0&0&0&0&0 \\
0&0&0&0&1&0&0&0&0&1&0&0&0&0&0&0&0&0&0&0&0&0&0 \\
0&0&1&0&0&0&0&0&0&0&0&0&1&0&0&0&0&0&0&0&0&0&0 \\
0&0&0&0&0&0&0&0&0&0&0&0&0&0&0&1&1&0&0&0&0&0&0 \\
0&0&0&0&0&0&0&0&0&0&1&1&0&0&0&0&0&0&0&0&0&0&0 \\
0&0&0&0&0&0&0&0&0&0&1&0&0&0&0&0&0&0&0&0&1&0&0 \\
0&0&0&1&1&0&0&0&0&0&0&0&0&0&0&0&0&0&0&0&0&0&0 
\end{array}
\right] . 
}
$$

\begin{eqnarray*}
v &=& [1,0,0,0,0,0,0,0,1,0,0,0,1,0,0,0,0,0,0,0,0,1,0] \\
w &=& [0,1,1,1,1,1,1,1,0,0,0,0,0,0,0,0,0,0,0,0,0,0,0]^T
\end{eqnarray*}

\bigskip

Step 3:  Now each of the identities in (\ref{idl}) corresponds to
a certain identity in matrices.  For example, the identity
$f(16n)  =  -2f(4n) + 3f(8n)$ can be written as
\begin{equation}
v M M_0 M_0 M_0 M_0 w = -2 v M M_0 M_0 w + 3 v M M_0 M_0 M_0 w,
\label{id16}
\end{equation}
where $M$ is the matrix product corresponding to the base-$2$
expansion of $n$.  More generally, we can think of $M$ as
some arbitrary product of the matrices $M_0$ and $M_1$,
starting with at least one $M_1$; this corresponds to an arbitrary
$n \geq 1$.
We can think of $M$ as a matrix of indeterminates.
Then (\ref{id16}) represents an assertion about the entries of
$M$ which can be verified.  Of course, the entries of $M$ are not
completely arbitrary, since they come about as $M_1$ times some product
of $M_0$ and $M_1$.   We can compute the (positive) transitive closure
of $M_0 + M_1$ and then multiply on the left by $M_1$; the entries that
have $0$'s will be $0$ in any product of $M_1$ times a product of
the matrices $M_0$ and $M_1$.  Thus we can replace the corresponding
indeterminates by $0$, which makes verifying (\ref{id16}) easier.

Another approach, which is even simpler, is to consider $vM$ in place
of $M$.  This reduces the number of entries it is required to check
from $d^2$ to $d$, where $d$ is the dimension of the matrices.

\bigskip

Step 4:  Finally, we have to verify the identies for $n = 0$ and $n =1 $,
which is easy.

\bigskip

We carried out this computation in Maple
for the matrices $M_0$ and $M_1$ corresponding
to $A$, which completes the proof.  The Maple program can be downloaded
from \newline
\centerline{\url{http://www.cs.uwaterloo.ca/~shallit/papers.html} . }
\end{proof}

\section{Determining the relations}
\label{kreg}

The verification method of the previous section can be extended to a method
to mechanically {\it find\/} the relations for {\it any} given $k$-regular
sequence $g$
(instead of guessing them and verifying them), given the linear
representation of $g$.

Suppose we are given the linear representation of a $k$-regular sequence $g$,
that is, vectors $v, w$ and matrices $M_0, M_1, \ldots, M_{k-1}$ such
that $g(n) = v M_{a_1} M_{a_2} \cdots M_{a_j} w$, where
$a_1 a_2 \cdots a_j = (n)_k$.


Now let $M$ be arbitrary and consider $vM$ as a vector with variable
entries, say $[a_1, a_2,\ldots, a_d]$.  Successively compute 
$vM M_y w$ for words $y$ of length $0, 1, 2, \ldots$ over
$\Sigma_k = \lbrace 0, 1, \ldots, k-1 \rbrace$; this will give an
expression in terms of the variables $a_1, \ldots, a_d$.   After at
most $d+1$ such relations, we find an expression for
$v M M_y w$ for some $y$ as a linear combination of previously
computed expressions.  When this happens, you no longer need to consider
any expression having $y$ as a suffix.  Eventually the procedure halts,
and this corresponds to a system of
equations like that in (\ref{unbtm}).  

Consider the following example.  Let $k =2$,
$v = [6,1]$, $w = [2,4]^T$, and
\begin{eqnarray*} 
M_0 &=& \left[ \begin{array}{cc}
	-3 & 1 \\
	1 & 4 
	\end{array} \right] \\
M_1 &=& \left[ \begin{array}{cc}
	0 & 2 \\
	-3 & 1 
	\end{array} \right]
\end{eqnarray*}

Suppose $M$ is some product of $M_0$ and $M_1$, and 
suppose $vM = [a, b]$.

We find
\begin{eqnarray*}
vMw &=& 2a + 4b \\
vMM_0 w &=& -2a + 18b \\
v M M_1 w &=& -8a -2b \\
v M M_0 M_0 w &=& 24a + 70b \\
v M M_1 M_0 w &=& 36a + 24  b
\end{eqnarray*}

and, solving the linear systems, we get
\begin{eqnarray*}
v M M_1 w &=& {{35} \over {11}} v Mw  - {{9} \over {11}} v M_0 w \\
v M M_0 M_0 w &=& 13 v Mw + v M_0 w \\
v M M_1 M_0 w &=& {{174} \over {11}} vMw - {{24} \over {11}} v M_0 w  .
\end{eqnarray*}

This gives us
\begin{eqnarray*}
g(2n+1) &=& {{35} \over 11} g(n) + {9 \over {11}} g(2n) \\
g(4n) &=& 13 g(n) + g(2n) \\
g(4n+2) &=& {{174} \over {11}} g(n) - {{24} \over {11}} g(2n) 
\end{eqnarray*}
for $n \geq 1$.


\section{The growth rate of $f(n)$}
\label{growth}

We now return to $f(n)$, the number of unbordered factors of
$\bf t$ of length $n$.
Here is a brief table of $f(n)$:

\begin{table}[H]
\begin{center}
\begin{tabular}{|c|c|c|c|c|c|c|c|c|c|c|c|c|c|c|c|c|c|c|c|c|c|c|c|c|c|c|c|c|c|c|c
|c|c|c|c|c|c|c|c|c|c|c|c|c|c|c|c|c|}
\hline
$n$   & 0 &  1 & 2 & 3 & 4 & 5 & 6 & 7 & 8 & 9 & 10 & 11 & 12 & 13 & 14 & 15 
 & 16 & 17 & 18 & 19 & 20 & 21 & 22 & 23 & 24 & 25 & 26 & 27 & 28 & 29 & 30 & 31 \\
\hline
$f(n)$& 1 &  2 & 2 & 4 & 2 & 4 & 6 & 0 & 4 & 4 & 4 & 4 & 12 & 0 & 4 & 4 
 & 8 & 4 & 8 & 0 & 8 & 4 & 4 & 8 & 24 & 0 & 4 & 4 & 8 & 4 & 8 & 4 \\
\hline
\end{tabular}
\end{center}
\end{table}

Kalle Saari (personal communication) asked about the growth rate of
$f(n)$.    The following results characterizes it.

\begin{theorem}
We have $f(n) \leq n$ for $n \geq 4$. Furthermore,
$f(n) = n$ infinitely often.    Thus,
$\limsup_{n \geq 1} f(n)/n = 1$.
\end{theorem}

\begin{proof}

We start by verifying the following relations:

\begin{eqnarray}
f(4n) &=& 2 f(2n), \quad\quad (n \geq 2)  \label{eq2} \\
f(4n+1) &=& f(2n+1), \quad\quad (n \geq 0) \label{eq3} \\
f(8n+2) &=& f(2n+1) + f(4n+3), \quad\quad (n \geq 1) \label{eq4} \\
f(8n+3) &=& -f(2n+1) + f(4n+2) \quad\quad (n \geq 2) \label{eq5} \\
f(8n+6) &=& -f(2n+1) + f(4n+2) + f(4n+3) \quad\quad (n \geq 2) \label{eq6}\\
f(8n+7) &=& 2f(2n+1) + f(4n+3) \quad\quad (n \geq 3) \label{eq7}
\end{eqnarray}

These can be verified in exactly the
same way that we verified the system (\ref{unbtm}) earlier.

We now verify, by induction on $n$, that $f(n) \leq n$ for $n \geq 4$.
The base case is $n = 4$, and $f(4) = 2$.  Now assume $n \geq 5$.
Otherwise,
\begin{itemize}
\item If $n \equiv \amod{0} {4}$, say $n = 4m$ and $m \geq 2$.
Then $f(4m) = 2 f(2m) \leq 2 \cdot 2m \leq 4m$ by
(\ref{eq2}) and induction.
\medskip

\item If $n \equiv \amod{1} {4}$, say $n = 4m+1$ for $m \geq 1$,
then $f(4m+1) = f(2m+1)$ by (\ref{eq3}).   But $f(2m+1) \leq 2m+1$
by induction for $m \geq 2$.  The case $m = 1$ corresponds
to $f(5) = 4 \leq 5$.
\medskip

\item If $n \equiv \amod{2} {8}$, say $n = 8m+2$, then for $m \geq 2$
we have $f(8m+2) = f(2m+1) + f(4m+3) \leq 6m+4$ by induction, which is less than
$8m+2$.  If $m = 1$, then $f(10) = 4 < 10$.
\medskip

\item If $n \equiv \amod{3} {8}$, say $n = 8m+3$ for $m \geq 1$, then
$f(8m+3) = -f(2m+1) + f(4m+2) \leq f(4m+2) \leq 4m+2$ by induction.
\medskip

\item If $n \equiv \amod{6} {8}$, say $n = 8m+6$, then
$f(8m+6) = -f(2m+1) + f(4m+2) + f(4m+3) \leq f(4m+2) + f(4m+3) \leq 8m+5$
by induction,
provided $m \geq 2$.  For $m = 0$ we have $f(6) = 6$ and for
$m = 1$ we have $f(14) = 4$.
\medskip

\item If $n \equiv \amod{7} {8}$, say $n = 8m+7$, then
$f(8m+7) = 2f(2m+1) + f(4m+3) \leq 2 (2m+1) + 4m+3 = 8m + 5$ for
$m \geq 3$, by induction.
The cases $m = 0, 1, 2$ can be verified by inspection.
\end{itemize}

This completes the proof that $f(n) \leq n$.

It remains to see that $f(n) = n$ infinitely often.  We do this
by showing that
$f(n) = n$ for $n$ of the form $3 \cdot 2^i$, $i \geq 1$.  
Let us prove this by induction on $i$.  It is true for $i = 1$ since
$f(6) = 6$.    Otherwise $i \geq 2$,
and using (\ref{eq2}) we have $f(3 \cdot 2^{i+1}) = 
2f(3 \cdot 2^i) = 2 \cdot 3 \cdot 2^i = 3 \cdot 2^{i+1}$ by induction.
This also implies the claim $\limsup_{n \geq 1} f(n)/n = 1$.

\end{proof}

\section{Unbordered factors of other sequences}
\label{other}

We can carry out similar computations for other famous sequences.
In some cases the automata and the corresponding matrices are very 
large, which renders the computations time-consuming and the asymptotic
behavior less transparent.    We report on some of these computations,
omitting the details.

\begin{theorem}
Let ${\bf r} = r_0 r_1 r_2 \cdots = 00010010 \cdots $
denote the Rudin-Shapiro sequence,
defined by $r_n = $ the number of occurrences, taken modulo $2$, of
`11' in the binary expansion of $n$.  
Let $f_{\bf r} (n)$ denote the number of unbordered factors
of length $n$ in $\bf r$.
Then
$f_{\bf r}(n) \leq {{21}\over 8} n$ for all 
$n \geq 1$.  
Furthermore if $n = 2^i + 1$, then
$f(n) = 21 \cdot 2^{i-3}$ for $i \geq 4$.
\end{theorem}

\begin{theorem}
Let ${\bf p} = p_0 p_1 p_2 \cdots = 0100 \cdots$ be the
so-called ``period-doubling'' sequence, defined by 
$$
p_n =  \begin{cases}
1, & \text{ if $t_n = t_{n+1}$;} \\
0, & \text{ otherwise,}
\end{cases}
$$
where $t_0 t_1 t_2 \cdots$ is the Thue-Morse word $\bf t$.
Note that $\bf p$ is the fixed point of the morphism
$0 \rightarrow 01$ and
$ 1 \rightarrow 00$.
Then $f_{\bf p} (n)$, the number of unbordered factors of $\bf p$ of
length $n$, is equal to $2$ for all $n \geq 1$.
\end{theorem}

The period-doubling sequence can be generalized to base $k \geq 2$, as follows:
$${\bf p}_k := (\nu_k (n+1) \bmod 2)_{n \geq 0},$$
where $\nu_k(x)$ is the exponent of the largest power of $k$ dividing
$x$.  For each $k$, the corresponding sequence ${\bf p}_k$ is a binary
sequence that is $k$-automatic:

\begin{theorem}
Let $k$ be an integer $\geq 2$.
The sequence ${\bf p}_k$ is the fixed point of the morphism
$\varphi_k$, where
\begin{eqnarray*}
\varphi_k (0) &=& 0^{k-1}\, 1 \\
\varphi_k (1) &=& 0^k .
\end{eqnarray*}
\end{theorem}

\begin{proof}
Note that ${\bf p}_k (n) = c$ iff
$\nu_k (n+1) = 2j + c$ for some integer $j \geq 0$, and $c \in \lbrace 0,1
\rbrace$.

If $0 \leq a < k-1$, then
${\bf p}_k (kn+a) = \nu_k (kn+a+1) \bmod 2 = 0$.
If $a = k-1$ we have
${\bf p}_k (kn+a) = \nu_k (kn+k) \bmod 2 =
\nu_k(k(n+1)) \bmod 2 = (2j+c+1) \bmod 2$.
Hence if ${\bf p}_k(n) = 0$, then ${\bf p}_k[kn..kn+k-1] = 0^{k-1} \, 1$, while
if ${\bf p}_k (n) = 1$, then ${\bf p}_k[kn..kn+k-1] = 0^k$.
It follows that ${\bf p}_k$ is the fixed point of $\varphi_k$.
\end{proof}

The generalized sequence
${\bf p}_k$ has the same property of unbordered factors as
the period-doubling sequence:

\begin{theorem}
The number of unbordered factors of ${\bf p}_k$ of length $n$, for
$k \geq 2$ and $n \geq 1$, is equal to $2$, and the two unbordered
factors are reversals of each other.
\label{unbord}
\end{theorem}

We begin with some useful lemmas.

\begin{lemma}
Let $x \in \lbrace 0, 1 \rbrace^*$ be a word.  Then 
$0^{k-1} \, \varphi_k (x)^R = \varphi_k (x^R) \, 0^{k-1}$.
\label{rever}
\end{lemma}

\begin{proof}
Suppose $x = a_1 a_2 \cdots a_n$, where each $a_i \in \lbrace 0, 1 \rbrace$.
If $a \in \lbrace 0, 1 \rbrace$, let $\overline{a}$ denote $1-a$.
Then
\begin{eqnarray*}
0^{k-1} \, \varphi_k (x)^R &=&  0^{k-1} \, \left( \prod_{1 \leq i \leq n}
	\varphi_k (a_i) \right)^R \\
&=& 0^{k-1} \, \left( \prod_{1 \leq i \leq n} 0^{k-1}  \ 
	\overline{a_i} \ \right)^R  \\
&=& 0^{k-1} \left( \prod_{1 \leq i \leq n} \overline{a_{n+1-i}} \ 0^{k-1} \right) \\
&=& \left( \prod_{1 \leq i \leq n} 0^{k-1} \ \overline{a_{n+1-i}} \right) 0^{k-1} \\
&=& \left( \prod_{1 \leq i \leq n} \varphi_k (a_{n+1-i}) \right) 0^{k-1
} \\
&=& \varphi_k (x^R) \, 0^{k-1}.
\end{eqnarray*}
\end{proof}

\begin{lemma}
If the word $w$ is bordered, then $\varphi_k (w)$ is bordered.
\label{unbl}
\end{lemma}

\begin{proof}
If $w$ is bordered, then $w = xyx$ for $x \not= \epsilon$.  Then
$\varphi_k (w) = \varphi_k(x) \varphi_k (y) \varphi_k (x)$ is
bordered.
\end{proof}

\begin{lemma}
If $w$ is a factor of ${\bf p}_k$, then so is $w^R$.
\label{rev}
\end{lemma}

\begin{proof}
If $w$ is a factor of ${\bf p}_k$, then it is a factor of some
prefix ${\bf p}_k [0..k^i - 1]$ for some $i \geq 1$.  So it suffices
to show that ${\bf p}_k [0..k^i - 1]^R$ appears as a factor of
${\bf p}_k$.  In fact, we claim that 
$${\bf p}_k [0..k^i-1]^R = {\bf p}_k [k^i - 1..2k^i - 2].$$

To see this, it suffices to observe that $\nu_k (k^i - a) = \nu_k(k^i + a)$ for
$0 \leq a < k^i$.
\end{proof}

The following lemma describes the unbordered factors of $\varphi_k$.
If $w = 0^a x$, then by $0^{-a} \, w$ we mean the word $x$.

\begin{lemma}
\begin{itemize}
\item[(a)]  If $w$ is an unbordered factor of ${\bf p}_k$ and
$|w| \equiv \amod{0} {k}$, then $w = \varphi_k(x)$ or
$w = \varphi_k (x)^R$, for some unbordered factor $x$ of
${\bf p}_k$ with $|x| = |w|/k$.

\item[(b)]  If $w$ is an unbordered factor of ${\bf p}_k$ and
$|w| \equiv \amod{a} {k}$ for $0 < a < k$, then 
$w = 0^{a-k} \, \varphi_k (x)$ or
$w = \varphi_k (x)^R \, 0^{a-k}$, for some unbordered factor 
$x$ of ${\bf p}_k$ with $|x|= (|w|-a)/k + 1$.
\end{itemize}
\label{ham}
\end{lemma}

\begin{proof}
(a):  Suppose that $w = {\bf p}_k[i..i+kn-1]$ for some integer $i$.
There are two cases to consider:  ${\bf p}_k [i] = 0$ and 
${\bf p}_k [i] = 1$.

Suppose ${\bf p}_k [i] = 0$.  Since $\bf w$ is unbordered, we
have ${\bf p}_k [i+kn-1] = 1$.  Then $\nu_k (i+kn) \geq 1$, so
$i+kn = km$ for some $m \geq 0$.  Then $i = k(m-n)$ is a multiple of
$k$, so $w = \varphi_k (x)$, where $x = {\bf p}_k [i/k..i/k+n-1]$.   Note
that $|x| = |w|/k$.  Finally, Lemma~\ref{unbl} shows that
$x$ is unbordered.

Suppose ${\bf p}_k [i] = 1$.  Since $w$ is unbordered, we
have ${\bf p}_k[i+kn-1] = 0$.  From Lemma~\ref{rev} we know that
$w^R$ is also a factor of ${\bf p}_k$ (and also is unbordered).
Then from the previous paragraph, we see that
$w^R = \varphi_k (x)$ for some unbordered factor $x$ of
${\bf p}_k$, with $|x| = |w|/k $.
Then $w = \varphi_k (x)^R$, as desired.

\medskip

(b):      Suppose that $w = {\bf p}_k [i..i+kn+a-1]$ for
$0 < a < k$.  There are two cases to consider:  ${\bf p}_k [i] = 0$
and ${\bf p}_k [i] = 1$.

Suppose that ${\bf p}_k [i] = 0$.  Since $w$ is unbordered, we
know that ${\bf p}_k [i+kn+a-1] = 1$.  Then
$\nu_k (i+kn+a) \geq 1$, so $i + kn + a = km$ for some $m \geq 0$.
Then $i-(k-a) = k(m-n-1)$ is a multiple of $k$.  Hence
$$0^{k-a} \, w = {\bf p}_k [i-(k-a)..i+kn+a-1] = 
	\varphi_k( {\bf p}_k [ (i+a)/k - 1 ..(i+a)/k + n-1 ] ).$$
Let $x = {\bf p}_k [(i+a)/k - 1 .. (i+a)/k + n - 1]$.  Then
$w = 0^{a-k} \, \varphi_k (x)$, and $|x| = (|w|-a)/k + 1$.
If $x$ is bordered, then using Lemma~\ref{unbl} we have that
$0^{k-a} \, w$ has a border of length $\geq k$,
so $w$ has a border of length at least $a$, a contradiction.

Suppose that ${\bf p}_k [i] = 1$.  Since $w$ is unbordered,
we know that ${\bf p}_k [i+kn+a-1] = 0$.  Then by 
Lemma~\ref{rev} we know that $w^R$ is also an unbordered factor
of ${\bf p}_k$.  Then from the previous paragraph, we get that
$w^R = 0^{a-k} \, \varphi_k (x)$ for some unbordered factor $x$ of
${\bf p}_k$ where $|x|  = (|w|-a)/k + 1$.  So $w =
\varphi_k (x)^R \, 0^{a-k}$, as desired.

\end{proof}

\begin{lemma}
Let $x$ be a word and $w = 1x0$ be an unbordered
word.  Then $0^i \varphi_k (x0)$ is unbordered
for $1 \leq i \leq k$.  
\label{ham5}
\end{lemma}

\begin{proof}
If $i = k$ then $0^k \varphi_k(x0) = \varphi_k (1x0) =
\varphi_k (w)$.  Suppose $\varphi_k (w)$ is bordered;
then there exist $u \not= \epsilon$ and $v$ such that
$\varphi_k (w) = uvu$.  Since $\varphi_k(0) = 0^{k-1} 1$,
we know $u$ ends in $1$.  But since $u$ is a prefix
of $\varphi_k (w)$ that ends in $1$, it follows that
$|u| \equiv \amod{0} {k}$, and so $u$ is the image of
some word $r$ under $\varphi_k$.  Hence $w$ begins and
ends with $r$, a contradiction.

Now assume $1 \leq i < k$ and $0^i \varphi_k (x0)$ is
bordered.  Then there exist $u \not= \epsilon$ and
$v$ such that $0^i \varphi_k (x0) = uvu$; note that $u$ must
end in $1$.  It follows that
$$ \varphi_k (w) = \varphi_k (1x0) = 0^k \varphi_k (x0) 
	= 0^{k-i} (0^i \varphi_k (x0)) = 0^{k-i} uvu .$$
Since $0^{k-i} u$ and $0^{k-i} uvu$ both end in $1$ and
$0^{k-i} uvu = \varphi_k (w)$, we have $|vu| 
\equiv \amod{0} {k}$.   Hence $|u| \equiv \amod{i} {k}$.
It follows that $0^{k-i}uv$ ends in $0^k$, so
$0^{k-i} uvu = \varphi_k (w)$ begins and ends in
$0^{k-i} u$, a contradiction.
\end{proof}

We are now ready for the proof of Theorem~\ref{unbord}.

\begin{proof}
First, we show that there is at least one unbordered factor of
every length, by induction on $n$.  The base cases are
$n < 2k$, and are left to the reader.
Otherwise $n \geq 2k$.  Write $n = kn' + i$ where
$1 \leq i \leq k$.  
By induction
there is an unbordered word $w$ of length $n' + 1$.
Using Lemma~\ref{rev}, we can assume that $w$ begins with $1$
and ends with $0$, say $w = 1x0$.  By Lemma~\ref{ham5} we have
that $0^i \varphi_k (x0)$ is unbordered, and it is of length
$i + kn' = n$.

It remains to prove there are exactly $2$ unbordered factors
of every length.

If $n \leq 2k$, then it is easy to see that the only unbordered factors
are $1 \, 0^{n-1}$ and $0^{n-1} \, 1$.

Now assume $n > 2k$ and that
there are only two unbordered factors of length $n'$ for
all $n' < n$; we prove it for $n$.
Let $w$ be an unbordered factor of length $n$.

If $n \equiv \amod{0} {k}$, then by Lemma~\ref{ham} (a), 
we know that either $w = \phi_k (x)$ or
$w = \phi_k (x)^R$, where $x$ is an unbordered factor of
length $n/k$.    By induction there are exactly $2$ unbordered
factors of length $n/k$; by Lemma~\ref{rev} they are reverses of
each other.
Let $x$ be such an unbordered factor; since $|x| = n/k > 2$,
either $x$ begins with $0$ and ends with $1$, or
begins with $1$ and ends with $0$.  In the former case,
the image $w = \varphi_k (x)$ begins and ends with $0$, a contradiction.
So $x$ begins with $1$ and ends with $0$.    But there is only
one such factor, so there are only two possibilities for $w$.

Otherwise let $a = n \bmod k$; then $0 < a < k$.
By Lemma~\ref{ham} (b),
we know that $w = 0^{a-k} \, \varphi_k (x)$ or
$w = \varphi_k (x)^R \, 0^{a-k}$, where $x$ is an unbordered
factor of length $(|w|-a)/k + 1 \geq 2$.    
By induction there are exactly $2$ such unbordered words;
by Lemma~\ref{rev} they are reverses of
each other.
Let $x$ be such an unbordered factor;
then either $x$ begins with $0$ and ends with $1$, or
begins with $1$ and ends with $0$.  Let us call them $x_0$ and $x_1$,
respectively, with $x_0 = x_1^R$.  Now
$\varphi_k(x_0)$ begins with $0^{k-1} \, 1$, and ends with $0^k$.
Hence, provided $a \not = 1$, we see that $w = 0^{a-k} \, \varphi_k(x_0)$ begins
with $0$ and ends with $0$, a contradiction.  
If $a = 1$, Lemma~\ref{ham} (b) gives the two factors
$0^{1-k} \, \varphi_k (x_0)$ and
$\varphi_k(x_0)^R \, 0^{1-k}$.  The former begins with $1$ and ends with
$0$; the latter begins with $0$ and ends with $1$.

In the latter case,
$x_1$ begins with $1$ and ends with $0$. There is only one such $x_1$
(by induction), and then either
$w = 0^{a-k} \, \varphi_k (x_1)$ or
$w = \varphi_k (x_1)^R \, 0^{a-k}$, giving at most two possibilities for $w$.
In the case $a = 1$, these two factors would seem to give a total of
four factors of length $n$.  However, there are only two, since
\begin{eqnarray*}
0^{1-k} \, \varphi_k (x_0) &=&
	0^{1-k} \, \varphi_k (x_1^R) = \varphi_k(x_1)^R \, 0^{1-k} \\
\varphi_k(x_0)^R \, 0^{1-k} &=&
	0^{1-k} \, \varphi_k (x_0^R) = 0^{1-k} \, \varphi_k (x_1) \\
\end{eqnarray*}
This completes the proof.
\end{proof}

\end{document}